\documentclass[onecolumn]{IEEEtran}

\usepackage[T1]{fontenc}
\usepackage{amsmath,amssymb}
\usepackage{newtxtext,newtxmath}
\usepackage{cite}
\usepackage{microtype}
\usepackage[hidelinks]{hyperref}
\hypersetup{
  pdftitle={Exact Minimum Distance of the Ding--Li--Xia Cyclic Codes},
  pdfauthor={Yutong Zhang}
}

\allowdisplaybreaks[2]

\newtheorem{lemma}{Lemma}
\newtheorem{theorem}{Theorem}
\newtheorem{definition}{Definition}

\newcommand{\F}{\mathbb{F}}
\newcommand{\wtq}{\operatorname{wt}_{q}}
\newcommand{\PG}{\operatorname{PG}}

\title{Exact Minimum Distance of the Ding--Li--Xia Cyclic Codes}

\author{Yutong~Zhang,~\IEEEmembership{Member,~IEEE,}
        and~Yaoran~Yang,~\IEEEmembership{Member,~IEEE}%
\thanks{Manuscript received Month Day, 2026; revised Month Day, 2026.}
\thanks{Yutong Zhang and Yaoran Yang are with the School of Mathematics,
Sichuan University, Chengdu 610065, China
(e-mail: yutongzhang@stu.scu.edu.cn; yangyaoran@stu.scu.edu.cn).}
\thanks{Corresponding author: Yutong Zhang.}}

\begin{document}

\maketitle

\begin{abstract}
The cyclic codes $\mho(q,m,h)$ introduced by Ding, Li, and Xia form a nonbinary generalization of punctured binary Reed--Muller codes. Ding, Li, and Xia established the bounds $(q^{h+1}-1)/(q-1)\leq d(\mho(q,m,h))\leq 2q^h-1$ and asked whether the BCH lower bound is always exact. This paper proves that, for every prime power $q$, every $m\geq 2$, and every $1\leq h\leq m-1$, the minimum distance is $d(\mho(q,m,h))=(q^{h+1}-1)/(q-1)$. The upper bound is obtained by an explicit projective-subspace construction. For any $(h+1)$-dimensional $\F_q$-subspace $V$ of $\F_{q^m}$, the set $V^{[q-1]}=\{x^{q-1}:x\in V\setminus\{0\}\}$ supports a codeword of weight $(q^{h+1}-1)/(q-1)$. Its membership in $\mho(q,m,h)$ follows from a vanishing lemma for subspace power sums and the digit-sum estimate $s_q((q-1)a)\leq(q-1)\wtq(a)$. The constructed codeword meets the BCH lower bound and therefore determines the exact minimum distance.
\end{abstract}


\section{Introduction}

Reed--Muller codes are among the foundational algebraic families in coding theory. They were introduced independently by Muller and Reed in 1954 through Boolean-function and decoding viewpoints, respectively \cite{Muller1954,Reed1954}. Their polynomial structure, recursive constructions, and large automorphism groups have made them central in error correction and in subsequent developments involving finite geometry, Boolean functions, and local testing. Although the original binary Reed--Muller codes are not cyclic, puncturing the coordinate indexed by the zero vector produces cyclic codes of length $2^m-1$.

The nonbinary theory developed along several related lines. Kasami, Lin, and Peterson introduced primitive generalized Reed--Muller codes and gave a cyclic description in terms of the base-$q$ digit sum of the defining exponents \cite{KasamiLinPeterson1968}. Delsarte, Goethals, and MacWilliams subsequently placed generalized Reed--Muller codes and their relatives in a unified framework and determined their fundamental parameters and minimum-weight structure \cite{DelsarteGoethalsMacWilliams1970}. A complementary projective branch was developed by Lachaud and S\o rensen, who studied evaluation codes on projective space and determined the parameters of projective Reed--Muller codes \cite{Lachaud1990,Sorensen1991}. These affine, punctured, and projective versions share the principle that restrictions on finite-field polynomial degrees or on $q$-adic digits produce highly structured code families whose low-weight words often reflect affine or projective subspaces.

The family considered in this paper is a different nonbinary extension of the punctured binary Reed--Muller codes. Ding, Li, and Xia introduced the cyclic codes $\mho(q,m,h)$ and related extended and reversible codes in order to construct linear complementary-dual codes and combinatorial designs \cite{DingLiXia2018}. Their defining set is governed not by the sum of the base-$q$ digits, as in the classical punctured generalized Reed--Muller construction, but by the number of nonzero base-$q$ digits. These two statistics coincide when $q=2$, so the binary members recover punctured binary Reed--Muller codes; for $q>2$, they lead to genuinely different cyclic codes. Ding, Li, and Xia also proved that the corresponding extended codes are affine-invariant and that their codewords support $2$-designs, thereby connecting the family with both affine geometry and the theory of codes with complementary duals.

More precisely, Ding, Li, and Xia determined the dimension of $\mho(q,m,h)$ and established $(q^{h+1}-1)/(q-1)\leq d(\mho(q,m,h))\leq 2q^h-1$. The lower bound follows from a consecutive-zero argument and the BCH bound, whereas the upper bound follows from an inclusion involving a punctured generalized Reed--Muller code \cite{DingLiXia2018}. The two bounds coincide in the binary case, and exact values were also obtained in several special nonbinary cases, including $h=m-1$ and $(q,h)=(3,1)$. Ding, Li, and Xia therefore posed the natural question of whether the BCH lower bound is the exact minimum distance for all admissible parameters.

Hu and Feng subsequently studied the minimum distances of both $\mho(q,m,h)$ and its reversible companion \cite{HuFeng2020}. Their method constructs sparse codewords from suitable divisors of $q^m-1$. It yields improved upper bounds in a number of parameter ranges and proves, in particular, that the lower bound is attained when $q\geq 3$ and $h+1$ divides $m$. They also obtained more detailed results for $h=1$. These results supplied substantial evidence for the conjectured formula, but the divisor condition is arithmetic in nature and does not cover arbitrary pairs $(m,h)$.

The present paper settles the minimum-distance problem for the full family. The argument is geometric and does not require a divisibility relation between $m$ and $h+1$. Given any $(h+1)$-dimensional $\F_q$-subspace $V\leq\F_{q^m}$, consider the image $V^{[q-1]}=\{x^{q-1}:x\in V\setminus\{0\}\}$. The fibers of $x\mapsto x^{q-1}$ are precisely the nonzero vectors on the one-dimensional $\F_q$-subspaces of $V$. Thus $V^{[q-1]}$ is naturally indexed by the projective points of $V$ and has cardinality $(q^{h+1}-1)/(q-1)$. We prove that its incidence vector, indexed by $\F_{q^m}^{*}$, is a codeword of $\mho(q,m,h)$.

Two elementary ingredients make the construction work. First, if $V$ has dimension $r$ and a positive integer $E$ satisfies $s_q(E)<r(q-1)$, then the subspace moment $\sum_{x\in V}x^E$ vanishes. Second, for every defining exponent $a$, the inequality $s_q((q-1)a)\leq(q-1)\wtq(a)$ converts the support weight of the base-$q$ expansion into the degree threshold needed by the moment-vanishing lemma. With $r=h+1$ and $\wtq(a)\leq h$, these facts imply the vanishing of all defining power sums over $V^{[q-1]}$. The resulting codeword has exactly the BCH lower-bound weight, which proves the claimed equality. In addition to resolving the open problem of Ding, Li, and Xia, the construction exhibits a projective-subspace family of minimum-weight codewords.

We now fix the notation used in the proof. Let $q=p^s$, where $p$ is prime, $s\geq 1$, $m\geq 2$, and $n=q^m-1$. Let $\F_q$ and $\F_{q^m}$ denote the finite fields with $q$ and $q^m$ elements, respectively. Choose $\alpha$ to be a generator of $\F_{q^m}^{*}$; thus $\F_{q^m}^{*}=\langle\alpha\rangle$ and $\operatorname{ord}(\alpha)=q^m-1=n$.

For an integer $N\geq 0$, write its base-$q$ expansion as $N=\sum_{j\geq 0}N_jq^j$, where $0\leq N_j\leq q-1$, and define its base-$q$ digit sum by $s_q(N)=\sum_{j\geq 0}N_j$. For $0\leq a\leq q^m-1$, write $a=\sum_{j=0}^{m-1}a_jq^j$, where $0\leq a_j\leq q-1$, and define its base-$q$ Hamming weight by
\[
 \wtq(a)=\bigl|\{j:0\leq j\leq m-1,\ a_j\neq 0\}\bigr|.
\]
The defining set of $\mho(q,m,h)$ is
\[
 I(q,m,h)=\{a:1\leq a\leq q^m-2,\ 1\leq\wtq(a)\leq h\},
\]
where $1\leq h\leq m-1$.

For completeness, this set is a union of $q$-cyclotomic cosets modulo $n$. Indeed, if $1\leq a\leq q^m-2$ and $a=\sum_{j=0}^{m-1}a_jq^j$, then the least nonnegative residue of $qa$ modulo $n=q^m-1$ is
\[
 a_{m-1}+\sum_{j=0}^{m-2}a_jq^{j+1},
\]
which is a cyclic shift of the base-$q$ digits. Hence $\wtq(qa\bmod n)=\wtq(a)$.

Equivalently, if $c(X)=\sum_{i=0}^{n-1}c_iX^i\in\F_q[X]/(X^n-1)$, then
\[
 c\in\mho(q,m,h)
 \quad\Longleftrightarrow\quad
 c(\alpha^a)=0\quad\text{for every }a\in I(q,m,h).
\]
Identifying coordinate $i$ with $\alpha^i\in\F_{q^m}^{*}$ and writing $c_\beta=c_i$ when $\beta=\alpha^i$, the same condition becomes
\[
 c\in\mho(q,m,h)
 \quad\Longleftrightarrow\quad
 \sum_{\beta\in\F_{q^m}^{*}}c_\beta\beta^a=0
 \quad\text{for every }a\in I(q,m,h).
\]
All power sums below are computed in $\F_{q^m}$. Section~II first recalls the BCH lower bound, then proves the two auxiliary lemmas, constructs the projective-subspace codeword, and concludes with the exact minimum-distance theorem.

\section{Main Theorem and Its Proof}

\subsection{BCH Lower Bound}

\begin{lemma}
For $m\geq 2$ and $1\leq h\leq m-1$, one has
\[
 d(\mho(q,m,h))\geq \delta_{q,h}:=\frac{q^{h+1}-1}{q-1}.
\]
\end{lemma}

\begin{IEEEproof}
Observe that $\delta_{q,h}=1+q+\cdots+q^h$ and $\delta_{q,h}-1=q+q^2+\cdots+q^h$. Since $h\leq m-1$, it follows that $\delta_{q,h}-1\leq q^m-2$.

If the base-$q$ representation of a positive integer has at least $h+1$ nonzero digits, then its smallest possible value is $1+q+\cdots+q^h=\delta_{q,h}$. Therefore, every integer satisfying $1\leq a\leq\delta_{q,h}-1$ has at most $h$ nonzero digits in its base-$q$ representation. Hence
\[
 \{1,2,\ldots,\delta_{q,h}-1\}\subseteq I(q,m,h).
\]
The generator polynomial thus has the consecutive zeros $\alpha,\alpha^2,\ldots,\alpha^{\delta_{q,h}-1}$. The BCH bound \cite{MacWilliamsSloane1977} gives $d(\mho(q,m,h))\geq\delta_{q,h}$.
\end{IEEEproof}

\subsection{Vanishing of Subspace Power Sums}

\begin{lemma}
Let $V$ be an $r$-dimensional $\F_q$-subspace of $\F_{q^m}$, where $1\leq r\leq m$. If a positive integer $E$ satisfies $s_q(E)<r(q-1)$, then
\[
 \sum_{x\in V}x^E=0.
\]
\end{lemma}

\begin{IEEEproof}
Choose an $\F_q$-basis $v_1,\ldots,v_r$ of $V$. Every $x\in V$ can be written uniquely as $x=z_1v_1+z_2v_2+\cdots+z_rv_r$, where $z_1,\ldots,z_r\in\F_q$. Write $E=\sum_{j\geq 0}E_jq^j$, where $0\leq E_j\leq q-1$. Since $z_i^{q^j}=z_i$ for every $z_i\in\F_q$, one has
\[
 x^E=\prod_{j\geq 0}\bigl(x^{q^j}\bigr)^{E_j}
 =\prod_{j\geq 0}\left(\sum_{i=1}^{r}z_iv_i^{q^j}\right)^{E_j}.
\]
The right-hand side is a polynomial in $z_1,\ldots,z_r$ whose total degree is at most $\sum_{j\geq 0}E_j=s_q(E)<r(q-1)$.

Expanding it and retaining only the monomials with nonzero coefficients gives
\[
 x^E=\sum_{\boldsymbol d}A_{\boldsymbol d}z_1^{d_1}\cdots z_r^{d_r},
 \qquad A_{\boldsymbol d}\in\F_{q^m}.
\]
Every exponent vector $\boldsymbol d=(d_1,\ldots,d_r)$ occurring here satisfies $d_1+\cdots+d_r<r(q-1)$. Consequently, for each such monomial, there is at least one index $i$ for which $0\leq d_i<q-1$.

We now show that this monomial sums to zero over $\F_q^r$. Indeed,
\[
 \sum_{(z_1,\ldots,z_r)\in\F_q^r}z_1^{d_1}\cdots z_r^{d_r}
 =\prod_{k=1}^{r}\left(\sum_{z\in\F_q}z^{d_k}\right).
\]
For the index $i$ chosen above, if $d_i=0$, then $\sum_{z\in\F_q}z^{d_i}=\sum_{z\in\F_q}1=q=0$ in characteristic $p$. If $1\leq d_i<q-1$, then
\[
 \sum_{z\in\F_q}z^{d_i}=\sum_{z\in\F_q^{*}}z^{d_i}=0,
\]
because $d_i$ is not a multiple of $q-1$ and $\F_q^{*}$ is cyclic of order $q-1$. Thus every monomial sums to zero over $\F_q^r$, and therefore $\sum_{x\in V}x^E=0$.
\end{IEEEproof}

\subsection{Digit-Sum Estimate}

\begin{lemma}
Let $1\leq a\leq q^m-2$, and write $a=\sum_{j=0}^{m-1}a_jq^j$, where $0\leq a_j\leq q-1$. Then
\[
 s_q((q-1)a)\leq(q-1)\wtq(a).
\]
\end{lemma}

\begin{IEEEproof}
Let $J=\{j:0\leq j\leq m-1,\ a_j\neq 0\}$. By the subadditivity of the base-$q$ digit sum, $s_q(A+B)\leq s_q(A)+s_q(B)$, because carrying can only decrease the digit sum. Hence
\[
 s_q((q-1)a)
 =s_q\!\left(\sum_{j\in J}(q-1)a_jq^j\right)
 \leq\sum_{j\in J}s_q((q-1)a_jq^j).
\]
Multiplication by $q^j$ merely appends zeros to the base-$q$ representation, so $s_q((q-1)a_jq^j)=s_q((q-1)a_j)$.

If $1\leq d\leq q-1$, then
\[
 (q-1)d=(d-1)q+(q-d),
\]
where $0\leq d-1\leq q-2$ and $1\leq q-d\leq q-1$. This is a valid base-$q$ representation, and $s_q((q-1)d)=(d-1)+(q-d)=q-1$. Consequently,
\[
 s_q((q-1)a)\leq\sum_{j\in J}(q-1)
 =(q-1)|J|=(q-1)\wtq(a).
\]
\end{IEEEproof}

\subsection{Projectivized Power Sums}

\begin{definition}
Let $V\leq\F_{q^m}$ be an $r$-dimensional $\F_q$-subspace. Define
\[
 V^{[q-1]}=\{x^{q-1}:x\in V\setminus\{0\}\}\subseteq\F_{q^m}^{*}.
\]
Also write
\[
 \PG(V)=\{L\leq V:\dim_{\F_q}L=1\},
\]
the projective point set consisting of all one-dimensional $\F_q$-subspaces of $V$.
\end{definition}

\begin{lemma}
Let $V\leq\F_{q^m}$ have $\F_q$-dimension $r=h+1$. Then
\[
 |V^{[q-1]}|=\frac{q^{h+1}-1}{q-1}.
\]
Moreover, for every integer $a$ satisfying $1\leq a\leq q^m-2$ and $\wtq(a)\leq h$, one has
\[
 \sum_{\beta\in V^{[q-1]}}\beta^a=0.
\]
\end{lemma}

\begin{IEEEproof}
For $x,y\in V\setminus\{0\}$, one has $x^{q-1}=y^{q-1}$ if and only if $(x/y)^{q-1}=1$. The complete set of roots of $T^{q-1}=1$ in $\F_{q^m}^{*}$ is precisely $\F_q^{*}$: the latter already contains $q-1$ roots, while $T^{q-1}-1$ has degree $q-1$. Therefore,
\[
 x^{q-1}=y^{q-1}\quad\Longleftrightarrow\quad x/y\in\F_q^{*}.
\]
The fibers of $x\mapsto x^{q-1}$ on $V\setminus\{0\}$ are thus exactly the sets obtained by removing the zero vector from one-dimensional $\F_q$-subspaces. Hence
\[
 |V^{[q-1]}|=\frac{|V|-1}{q-1}=\frac{q^{h+1}-1}{q-1}.
\]

Now let $a$ satisfy $1\leq a\leq q^m-2$ and $\wtq(a)\leq h$, and set $E=(q-1)a$. By the digit-sum estimate,
\[
 \begin{aligned}
 s_q(E)&=s_q((q-1)a)\leq(q-1)\wtq(a)\\
 &\leq h(q-1)<(h+1)(q-1).
 \end{aligned}
\]
Since $\dim_{\F_q}V=h+1$, the subspace power-sum vanishing lemma gives $\sum_{x\in V}x^E=0$. Because $E>0$, the term with $x=0$ is zero, so
\[
 \sum_{x\in V\setminus\{0\}}x^{(q-1)a}=0.
\]
Choose one representative $x_L\in L\setminus\{0\}$ from each $L\in\PG(V)$. Then
\[
 \begin{aligned}
 0
 &=\sum_{x\in V\setminus\{0\}}x^{(q-1)a}\\
 &=\sum_{L\in\PG(V)}\sum_{\lambda\in\F_q^{*}}(\lambda x_L)^{(q-1)a}\\
 &=\sum_{L\in\PG(V)}x_L^{(q-1)a}
   \sum_{\lambda\in\F_q^{*}}\lambda^{(q-1)a}\\
 &=(q-1)\sum_{L\in\PG(V)}(x_L^{q-1})^a,
 \end{aligned}
\]
where $\lambda^{(q-1)a}=1$ for $\lambda\in\F_q^{*}$. In $\F_{q^m}$, the element $q-1$ equals $-1$, so it is nonzero and can be canceled. As $L$ varies, the elements $x_L^{q-1}$ run through $V^{[q-1]}$ without repetition. Hence $\sum_{\beta\in V^{[q-1]}}\beta^a=0$.
\end{IEEEproof}

\subsection{Main Theorem}

\begin{theorem}
Let $q$ be a prime power, let $m\geq 2$, and let $1\leq h\leq m-1$. Then
\[
 d(\mho(q,m,h))=\frac{q^{h+1}-1}{q-1}.
\]
More specifically, for every $(h+1)$-dimensional $\F_q$-subspace $V\leq\F_{q^m}$, define
\[
 c_V(X)=\sum_{\substack{0\leq i\leq n-1\\ \alpha^i\in V^{[q-1]}}}X^i
 \in\F_q[X]/(X^n-1).
\]
Then $c_V\in\mho(q,m,h)$ and
\[
 \operatorname{wt}(c_V)=|V^{[q-1]}|=\frac{q^{h+1}-1}{q-1}.
\]
\end{theorem}

\begin{IEEEproof}
Choose any $(h+1)$-dimensional $\F_q$-subspace $V\leq\F_{q^m}$. Index the coordinates of $c_V$ by $\F_{q^m}^{*}=\langle\alpha\rangle$, and set
\[
 c_\beta=
 \begin{cases}
 1, & \beta\in V^{[q-1]},\\
 0, & \beta\notin V^{[q-1]}.
 \end{cases}
\]
Thus $c_V(X)=\sum_{i=0}^{n-1}c_{\alpha^i}X^i$.

For any $a\in I(q,m,h)$, one has $1\leq a\leq q^m-2$ and $\wtq(a)\leq h$. By the projectivized power-sum lemma,
\[
 \sum_{\beta\in V^{[q-1]}}\beta^a=0.
\]
On the other hand,
\[
 \begin{aligned}
 c_V(\alpha^a)
 &=\sum_{i=0}^{n-1}c_i(\alpha^a)^i
 =\sum_{\beta\in\F_{q^m}^{*}}c_\beta\beta^a\\
 &=\sum_{\beta\in V^{[q-1]}}\beta^a=0.
 \end{aligned}
\]
Therefore, $c_V$ satisfies all zero conditions defining $\mho(q,m,h)$, and hence $c_V\in\mho(q,m,h)$. The counting formula in the projectivized power-sum lemma gives
\[
 \operatorname{wt}(c_V)=|V^{[q-1]}|=\frac{q^{h+1}-1}{q-1},
\]
so
\[
 d(\mho(q,m,h))\leq\frac{q^{h+1}-1}{q-1}.
\]
The BCH lower bound gives the reverse inequality. Combining the two proves the result.
\end{IEEEproof}
\section*{Declaration of Generative AI and AI-Assisted Technologies in the Writing Process}
The author used GPT 5.5 Pro to enumerate possible cases, search for counterexamples, and generate preliminary proofs. All AI-generated suggestions were manually reviewed, modified, and independently verified by the author. No unverified AI-generated proof or mathematical claim was incorporated into the final manuscript.

\end{document}